\documentclass[11pt]{article}
\usepackage[inline,shortlabels]{enumitem}
\usepackage[margin=1in]{geometry}
\usepackage{float}
\usepackage{mathtools}
\usepackage{comment}
\usepackage{graphicx}
\usepackage{longtable}
\usepackage{color}
\usepackage{amssymb,amsmath,amsthm}
\usepackage{natbib}
\usepackage{multirow}
\usepackage{booktabs}
\usepackage{subcaption}
\usepackage{tikz}
\usepackage[english]{babel}
\usepackage{hyperref}
\usepackage{authblk}

\usetikzlibrary{arrows,positioning,decorations.markings}
\pgfdeclarelayer{background}
\pgfsetlayers{background,main}

\newcommand{\supp}{\mathop{\mathrm{supp}}}

\newtheorem{theorem}{Theorem}
\newtheorem{definition}{Definition}
\newtheorem{proposition}{Proposition}

\title{Identification and estimation of mediational effects of longitudinal modified treatment policies}

\author[1]{Brian Gilbert}
\author[2]{Katherine L. Hoffman}
\author[2]{Nicholas Williams} 
\author[2]{Kara E. Rudolph}
\author[3]{Edward J. Schenck}
\author[1]{Iván Díaz\thanks{Corresponding author:ivan.diaz@nyulangone.org}}
\affil[1]{Division of Biostatistics, Department of Population Health, New York University Grossman School of Medicine}
\affil[2]{Department of Epidemiology, Mailman School of Public Health, Columbia University}
\affil[3]{Department of Medicine, Division of Pulmonary and Critical Care Medicine, NewYork-Presbyterian Hospital/Weill Cornell Medical Center}

\date{\today}\usepackage{float}
\usepackage{mathtools}
\usepackage{subdepth}
\usepackage{comment}
\usepackage{xr}
\usepackage{graphicx}
\usepackage{longtable}
\usepackage{color}
\usepackage{mathtools}
\usepackage{amssymb,amsmath,amsthm}
\usepackage{multirow}
\usepackage[titletoc,title]{appendix}
\usepackage{bbm}
\usepackage{setspace}
\usepackage{dsfont}
\usepackage[OT1]{fontenc}
\usepackage{subcaption}
\usepackage{refcount}
\usepackage{booktabs}
\usepackage{xr}

\AtEndDocument{\refstepcounter{theorem}\label{finalthm}}
\AtEndDocument{\refstepcounter{equation}\label{finaleq}}

{
  \theoremstyle{definition}
  
}
{
  \theoremstyle{definition}
  \newtheorem{assumptioniden}{}
}

 \DeclareMathOperator{\var}{\mathsf{Var}}

\renewcommand{\P}{\mathsf{P}}

\newcommand{\Q}{\mathsf{Q}}

\renewcommand{\d}{\mathsf{d}}
\newcommand{\g}{\mathsf{g}}

\newcommand{\G}{\mathsf{G}}

\renewcommand{\S}{\mathsf{S}}
\newcommand{\C}{\mathsf{K}}
\renewcommand{\H}{\mathsf{H}}

\newcommand{\indep}{\mbox{$\perp\!\!\!\perp$}}

\newcommand{\D}{\mathsf{D}}

\newcommand{\one}{\mathds{1}}
 
\newcommand{\E}{\mathsf{E}}

\renewcommand{\P}{\mathsf{P}}

\newcommand{\1}{\mathbb{I}}

\usepackage{accents}
\usepackage{tikz}
\usetikzlibrary{arrows}

\newcommand{\ubar}[1]{\underaccent{\bar}{#1}}
\renewenvironment{proof}{{\it Proof }}{\qed \\}

\DeclarePairedDelimiterX{\norm}[1]{\lVert}{\rVert}{#1}

\usetikzlibrary{decorations.markings}
\usepackage[linesnumbered]{algorithm2e}
\usepackage{url}

\pgfdeclarelayer{background}
\pgfsetlayers{background,main}
\usetikzlibrary{arrows,positioning}
\tikzset{
>=stealth',
punkt/.style={
rectangle,
rounded corners,
draw=black, very thick,
text width=6.5em,
minimum height=2em,
text centered},
pil/.style={
->,
thick,
shorten <=2pt,
shorten >=2pt,}
}
\newcommand{\Vertex}[3]%
{\node[minimum width=0.6cm,inner sep=0.05cm] (#2) at (#1) {#3};
}
\newcommand{\Vertexr}[3]%
{\node[rectangle, draw, minimum width=0.6cm,inner sep=0.05cm] (#2) at (#1) {#2};
}
\newcommand{\ArrowR}[3]%
{ \begin{pgfonlayer}{background}
\draw[->,#3] (#1) to[bend right=30] (#2);
\end{pgfonlayer}
}
\newcommand{\ArrowLW}[3]%
{ \begin{pgfonlayer}{background}
\draw[->,#3] (#1) to[bend left=30] (#2);
\end{pgfonlayer}
}

\newcommand{\ArrowL}[3]%
{ \begin{pgfonlayer}{background}
    \draw[->,#3] (#1) to[bend left=45] (#2);
  \end{pgfonlayer}
}

\newcommand{\EdgeL}[3]%
{ \begin{pgfonlayer}{background}
\draw[dashed,#3] (#1) to[bend right=-45] (#2);
\end{pgfonlayer}
}

\newcommand{\Arrow}[3]%
{ \begin{pgfonlayer}{background}
\draw[->,#3] (#1) -- +(#2);
\end{pgfonlayer}
}

\allowdisplaybreaks
\pgfarrowsdeclare{arcs}{arcs}{...}
{
  \pgfsetdash{}{0pt} %
  \pgfsetroundjoin   %
  \pgfsetroundcap    %
  \pgfpathmoveto{\pgfpoint{-5pt}{5pt}}
  \pgfpatharc{180}{270}{5pt}
  \pgfpatharc{90}{180}{5pt}
  \pgfusepathqstroke
}
\newcommand{\ArrowB}[3]%
{ \begin{pgfonlayer}{background}
    \draw[|-arcs,line width=0.4mm,shorten <= 0.3cm,shorten >= 0.3cm,#3] (#1) -- +(#2);
  \end{pgfonlayer}
}

\newcommand{\titlepaper}{Identification and estimation of mediational effects of longitudinal modified treatment policies}

\begin{document}
\title{\titlepaper}
\maketitle\begin{abstract}
{We demonstrate a comprehensive semiparametric approach to causal mediation analysis, addressing the complexities inherent in settings with longitudinal and continuous treatments, confounders, and mediators. Our methodology utilizes a nonparametric structural equation model and a cross-fitted sequential regression technique based on doubly robust pseudo-outcomes, yielding an efficient, asymptotically normal estimator without relying on restrictive parametric modeling assumptions.
We are motivated by a recent scientific controversy regarding the effects of invasive mechanical ventilation (IMV) on the survival of COVID-19 patients, considering acute kidney injury (AKI) as a mediating factor. We highlight the possibility of ``inconsistent mediation," in which the direct and indirect effects of the exposure operate in opposite directions. We discuss the significance of mediation analysis for scientific understanding and its potential utility in treatment decisions.}
\end{abstract}

\section{Introduction}\label{sec:intro}

Causal mediation analysis, or the identification of the different pathways through which an exposure or treatment (we will use the terms ``exposure" and ``treatment" interchangeably) can operate, has seen increasing interest across the scientific literature \citep{nguyen2021clarifying}. Although the concept of mediation has been active for at least a century \citep{wright1921correlation}, mediation analysis has only relatively recently been incorporated in formal causal models such as the framework of potential outcomes or structural causal models, which allow for nonparametric perspectives, or the study of mediation effects without reference to a particular (e.g., linear) statistical model \citep{imai2010general}.

Even so, the formulation of an estimand of interest in a mediation analysis is not always a simple task. The so-called ``natural" mediation effects, where one considers an intervention that partially ``blocks" the effect of an exposure by assigning mediator values to their counterfactual values under a fixed exposure level in order to isolate ``direct" and ``indirect" effects, have been well-studied, and conditions for their identification have been given \citep{RobinsGreenland92, Pearl01}. However, those identification conditions are quite strong and unlikely to be satisfied in many instances of real data-generating processes \citep{avin2005identifiability, miles2015partial}. In particular, they require the assumption of ``cross-world counterfactual independence," where counterfactual outcomes downstream of mutually exclusive counterfactual events are independent, which cannot be tested even experimentally. To analyze mediation under more reasonable assumptions, alternative estimands based on stochastic draws from the mediator distribution have been proposed and labeled ``interventional effects" \citep{vanderweele2014effect}. 

Mediation analysis in longitudinal settings is particularly challenging; some existing approaches have relied on parametric model specification \citep{vanderweele2017mediation, tai2022causal}, which leads to a lack of robustness. \cite{zheng2017longitudinal} provided a robust machine-learning approach for the estimation of mediation effects with time-varying treatments, confounders, and mediators. This line of research was advanced, for example, by \cite{diaz2023efficient} which allowed the (interventional) direct effect to include pathways through intermediate confounders, thereby allowing the indirect effect to isolate precisely those operating through the mediator and not other variables.

However, even these strategies that make use of flexible models for estimation still rely on the notion of a ``static" treatment regime, or one in which each value of treatment for all timepoints is explicitly specified in the estimand of interest, which is a significant limitation, given that real-world decision-making rarely follows this form in the longitudinal setting. That is, even an externally designed intervention is likely to be responsive, intentionally or otherwise, to variables that are encountered over the course of treatment. A few examples are the administration of antiretroviral drugs in HIV patients, which can be adapted to CD4 cell count measurements \citep{hernan2006comparison}, and the administration of antihypertensive drugs, which can be adapted to blood pressure measurements \citep{johnson2018causal}.

Aside from their lack of applicability, studies of static interventions face various statistical challenges. First, identification of the causal estimand relies on the positivity assumption (i.e., that there is a positive probability of any treatment course given any level of covariates), which is often implausible, especially for continuous treatments given at multiple timepoints. Likewise, for continuous treatments, the estimation of a dose-response curve (the function of exposure values that yields the population-average outcome conditional on receiving each value) involves complications in terms of convergence rates \citep[standard $\sqrt{n}$ rates are not possible, see][]{kennedy2017nonparametric} and of the summarization and interpretation of such a high-dimensional target parameter.

By contrast, so-called ``dynamic" treatment regimes allow exposure values to change based on observed covariates \citep{robins2004effects}. These can be generalized to ``stochastic" interventions which allow for randomness in the exposure values, even conditional on all other variables. Furthermore, it may often be of interest to consider a treatment regime that assigns treatment values depending on what treatment value would have occurred in the absence of intervention, or the ``natural" value of treatment \citep{young2014identification}. Such a regime, which we call a ``modified treatment policy," was considered in \cite{robins2004effects}, with formalization and estimation strategies for single-timepoint studies to follow in \cite{Diaz12} and \cite{Haneuse2013}. \cite{diaz2021nonparametricmtp} provide a framework for robust estimation of the effects of modified treatment policies for time-varying treatments. 

We synthesize these various strands of the statistical literature into a new strategy for mediation analysis in longitudinal settings, following \cite{diaz2020causal} and \cite{hejazi2023nonparametric}, who studied mediation analysis for single-timepoint stochastic interventions. Specifically, we consider the direct and indirect effects of longitudinal modified treatment policies under interventional mediator distributions, which have not been considered in previous literature. The methods we propose are motivated by an issue that was relevant in the early landscape of COVID-19 management: the use of invasive mechanical ventilation (IMV) in cases of acute respiratory distress syndrome (ARDS). ARDS is a form of acute hypoxic respiratory failure, which is one of the key time-varying features of a patient's course with severe COVID (see, e.g., \cite{wang2021cardiovascular}). 
In the beginning stages of the pandemic, IMV was employed early, or at a lower supplemental oxygen delivery threshold, to manage ARDS, but its usage has been associated with several iatrogenic risks such as ventilator-associated pneumonia \citep{wicky2021ventilator} and barotrauma \citep{shrestha2022pulmonary}. Of interest to the present study is acute kidney injury (AKI), a critical condition that complicates ICU stays and is associated with mortality \citep{vemuri2022association}. As more data emerged and practitioners became more familiar with the pace of disease progression, guidelines shifted towards delaying intubation due to high mortality rates among mechanically ventilated patients and potential secondary complications. To navigate these treatment dynamics, we aim to elucidate the direct and indirect effects of IMV on patient mortality with respect to the AKI pathway. This application offers a compelling case for the use of our methods in understanding the effects of time-varying treatments and the causal pathways involved. 
We emphasize the phenomenon of ``inconsistent mediation" (a situation in which the direct and indirect effects are in opposite directions). In the current setting, this would involve (on average) IMV causing death through AKI while (on average) \textit{preventing} death through other mechanisms. This dichotomy may be relevant for research and clinical practice, if interventions could be in place to mitigate IMV-induced AKI. While this is an active area of research, potential interventions include the administration of diuretics to modulate intravascular volume status \citep{grams2011fluid, glassford2011fluid} and pharmacological options as explored by \cite{pickkers2022new}. Finally, we also investigate the relationship between baseline covariates and the magnitude of these direct and indirect effects; such relationships can suggest biological pathways as well as potentially inform future treatments.

 The current paper is structured as follows. Section \ref{sec:setup} introduces the causal framework and notation we will use throughout. Section \ref{sec:learning} presents the theory and techniques we will use to learn mediational effects of longitudinal treatment policies. Section \ref{sec:aplica} gives our analysis of the intubation data, and Section \ref{sec:discussion} concludes.

\section{Model setup}\label{sec:setup}

\subsection{Notation}\label{ssec:notation}

Our notation largely mirrors that of \cite{diaz2021nonparametricmtp} and \cite{diaz2023efficient}. We consider data $X_1,\ldots, X_n$, an i.i.d. sample from a distribution $\P$. The data are longitudinal, with $X=(L_1, A_1, Z_1, M_1, L_2, \ldots, A_\tau, Z_\tau, \allowbreak
M_\tau, L_{\tau+1})$, where the variables $A$ are the exposure of interest, $L$ are covariates, $M$ are mediators, and $Z$ are intermediate confounders. Under right censoring, we let $A_t$ denote a bivariate variable that includes an indicator of censoring as well as the exposure, and let the data become degenerate for any timepoint after the censoring time, adopting the methodology described in \cite{diaz2022causalcomp}. An intermediate confounder is a variable affected by exposure and which affects the mediator and outcome; see \cite{vanderweele2014effect}. We let $Y=L_{\tau+1}$ denote the outcome of interest. In our illustrative application, $A$ will denote intubation status, $M$ is the occurrence of AKI, $Y$ is mortality/survival, and $L$ and $Z$ are covariates or lab results, either at baseline or during the hospital stay. As in our illustrative application, we will assume that the mediator takes values in a finite set; the supports of all other variables are unrestricted.

For any symbol $W$, we let $\bar W_t = (W_1, \ldots, W_t)$ (past history) and $\ubar W_t = (W_t, \ldots, W_\tau)$ (future history). We let $\bar W = \bar W_{\tau} = \ubar W_1$, the full history. We also define $H_{A,t} = (\bar L_t, \bar M_{t-1}, \bar Z_{t-1}, \bar A_{t-1})$; this is the history of all variables prior to $A_t$. Similarly, $H_{Z,t}$ is the history of all variables prior to $Z_t$, and $H_{M,t}$ and $H_{L,t}$ are defined likewise. Finally,  $\g_{A,t}(a_t \mid h_{A,t})$ is the probability density or probability mass function of treatment at time $t$, given the variable history $H_{A,t}=h_{A,t}$. Similarly, $\g_{M,t}(m_t, h_{M,t})$ gives the conditional probability mass function of the time-varying mediator.%

\subsection{Structural equation model and definition of causal effects}

We consider causal effects defined within a nonparametric structural equation model, which consists of exogenous and endogenous variables and unknown deterministic functions governing the causal relations between them, together with independence assumptions on the exogenous variables \citep{Pearl00}. We allow for arbitrary causal relations among observed variables that respect their time-ordering; this can be visualized in the saturated directed acyclic graph (DAG) given in Figure \ref{fig:dag}, reproduced from \cite{diaz2023efficient}.
 
\begin{figure}[!htb]
  \centering
  \begin{tikzpicture}
    \Vertex{-1, 0}{W}{$L_1$}
    \Vertex{1, 0}{A1}{$A_1$}
    \Vertex{2, 0}{M1}{$M_1$}
    \Vertex{3, 0}{Z1}{$L_2$}
    \Vertex{2, -1}{L1}{$Z_1$}
    \ArrowB{W}{A1}{black}
    \Arrow{A1}{M1}{black}
    \Arrow{A1}{L1}{black}
    \Arrow{M1}{Z1}{black}
    \Arrow{L1}{M1}{black}
    \Arrow{L1}{Z1}{black}
    \ArrowL{A1}{Z1}{black}

    \Vertex{5, 0}{A2}{$A_2$}
    \Vertex{6, 0}{M2}{$M_2$}
    \Vertex{7, 0}{Z2}{$L_3$}
    \Vertex{6, -1}{L2}{$Z_2$}
    \node (dots) at (8, 0) {$\cdots$};
    \Arrow{A2}{M2}{black}
    \Arrow{A2}{L2}{black}
    \Arrow{M2}{Z2}{black}
    \Arrow{L2}{M2}{black}
    \Arrow{L2}{Z2}{black}
    \ArrowL{A2}{Z2}{black}

    \ArrowB{Z1}{A2}{black}

    \Vertex{9, 0}{At}{$A_\tau$}
    \Vertex{10, 0}{Mt}{$M_\tau$}
    \Vertex{11, 0}{Zt}{$Y$}
    \Vertex{10, -1}{Lt}{$Z_\tau$}
    \Arrow{At}{Mt}{black}
    \Arrow{At}{Lt}{black}
    \Arrow{Mt}{Zt}{black}
    \Arrow{Lt}{Mt}{black}
    \Arrow{Lt}{Zt}{black}
    \ArrowL{At}{Zt}{black}
  \end{tikzpicture}
  \caption{Causal DAG for the model of time-varying mediation. For clarity, the symbol
    $\boldsymbol{\mapsto}$ is used to indicate arrows from all previous nodes (to the left)
    to all following nodes (to the right).}
  \label{fig:dag}
\end{figure}
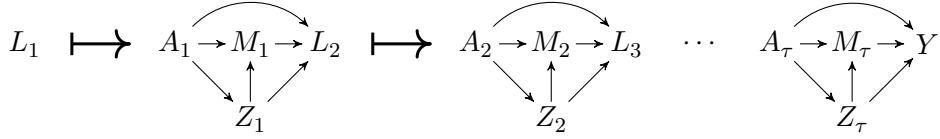

To formalize the dependence structure, we posit functions $f$ such that for each $t$, we have  $A_t=f_{A,t}(H_{A, t}, U_{A,t})$,
$Z_t=f_{Z,t}(H_{Z,t}, U_{Z,t})$, $M_t=f_{M,t}(H_{M,t}, U_{M,t})$, and
$L_t=f_{L,t}(H_{L,t}, U_{L,t})$. The variables $U=(U_{A,t}, U_{Z,t}, U_{M,t}, U_{L,t}, U_Y:t\in \{1,\ldots,\tau\})$ are unobserved ``exogenous" variables. In the general model, there are no restrictions on the distribution of $U$, but in Section \ref{sec:learning}, they will be subject to assumptions to enable identification of the estimands of interest. 

This setup allows us to speak of causal effects as the interventions on some functions $f$ while keeping the remaining structure of the model intact. For example, the simple intervention to set $A_1$ to the value $a_1$ replaces $f_{A,1}$ with the constant function $a_1$. In this paper, we are interested in the effects of interventions on the treatment variables $\bar A$, and we will analyze their decomposition through pathways that include mediating variables $\bar M$ and through pathways that do not. For this purpose, we will consider interventions on the functions $f_{A,t}$ and $f_{M,t}$ for $t=1,\ldots, \tau$.

For interventions on the treatment variable, we will consider longitudinal modified treatment policies (LMTPs) \citep{diaz2021nonparametricmtp}, which are hypothetical interventions defined in terms of a sequence
of functions $\d=(\d_1,\ldots,\d_\tau)$, where $\d_t$ is a function
$\d_t(a_t, h_{A,t}, \epsilon_t)$ that depends on a treatment value
$a_t$, a history value $h_{A,t}$, and possibly a randomizer
$\epsilon_t$. (We will not consider such a randomizer in this paper.) The intervention is defined by removing the equation
$A_t=f_{A_t}(H_{A,t}, U_{A,t})$ sequentially (from $t=1$ to $t=\tau$) from
the structural model, and replacing it with evaluations of the functions
$\d_t$ as follows. At the first timepoint, the LMTP assigns the
exposure as a new random variable
$A_1^\d = \d_1(A_1, H_{A,1}, \epsilon_1)$. This generates a
counterfactual treatment
$A_2(\bar A_{1}^\d) = f_{A_2}(H_{A,2}(\bar A_{1}^\d), U_{A,1})$ that
would have been observed if the intervention had been discontinued
right before $t=2$. Then, the intervention at $t=2$ assigns treatment
as
$A_2^\d = \d_2(A_2(\bar A_{1}^\d), H_{A,2}(\bar A_{1}^\d),
\epsilon_2)$. 

Across multiple timepoints, it is helpful to define the \textit{counterfactual} history $H_{A,t}(\bar A_{t-1}^\d)$ as the sequence of variables prior to $A_t$ that would have been observed under the LMTP given by $\bar A_{t-1}^\d$. 
In addition,
$A_t(\bar A_{t-1}^\d) = f_{A_t}(H_{A,t}(\bar A_{t-1}^\d), U_{A,t})$, termed the ``natural value of
  treatment" \citep{richardson2013single, young2014identification}, ``represents the value of treatment that would have been observed at
time $t$ under an intervention carried out up until time $t-1$ but
discontinued thereafter" \citep{diaz2021nonparametricmtp}. The intervention at time $t$ is given by
$A_t^\d = \d_t(A_t(\bar A_{t-1}^\d), H_{A,t}(\bar A_{t-1}^\d),
\epsilon_t)$.

The intervention on $\bar M$ is defined in terms of a stochastic
intervention. Specifically, let
$M_t(\bar A^\d)=f_{M,t}(A_t^\d, H_{M,t}(\bar A_t^\d), U_{M,t})$ denote
the mediator at time $t$ observed under the above LMTP, and let
$\bar J(\bar A^\d)$ denote a random draw from the distribution of
$\bar M(\bar A^\d)$, potentially conditioning on baseline variables
$L_1$. We consider interventions where $\bar M$ is replaced by this
random draw.

A causal ``effect" is typically conceived as some difference in outcome distributions under different interventions. To simplify notation throughout, we will consider two user-specified, fixed interventions $\d^\star=(\d^\star_1,\ldots, \d^\star_\tau)$ and
$\d'=(\d'_1,\ldots, \d'_\tau)$, though the generalization of the theory to arbitrary $\d$ is immediate. We define the randomized modified treatment policy
effect as
$\E[Y(\bar A^{\d'}, \bar J(\bar A^{\d'})) - Y(\bar A^{\d^\star}, \bar
J(\bar A^{\d^\star}))]$, where $Y(\bar A, \bar J)$ denotes the potential outcome of $Y$ where the treatment is intervened on to equal $\bar A$, and the mediator is intervened on to equal $\bar J$.
We decompose this into an interventional direct effect
\begin{equation}\label{eq:de}
 DE = \E[Y(\bar A^{\d'}, \bar J(\bar A^{\d^\star})) - Y(\bar A^{\d^\star}, \bar J(\bar A^{\d^\star}))]\end{equation}
 and indirect effect
 \begin{equation}\label{eq:ie}
     IE=\E[Y(\bar A^{\d'}, \bar J(\bar A^{\d'})) - Y(\bar A^{\d'}, \bar J(\bar A^{\d^\star}))],\end{equation} which sum to the ``overall" or ``total" effect comparing $\d'$ to $\d^\star$. Thus, the task at hand is to estimate expressions of the form $\E  [Y(\bar A^{\d'}, \bar J(\bar A^{\d^\star}))]$, for general $\d^\star$ and $\d'$.

For more details on the interpretation of the interventional effect, see Section \ref{ssec:interp} of the supplement, which also describes the nuance of the direct/indirect distinction for longitudinal data.

\section{Identification and estimation of interventional mediation effects}\label{sec:learning}
\subsection{Assumptions and sequential regression}

In what follows, it will be necessary to reference variable histories under interventions. For example, we let $H_{L,t}' =(\bar L_{t-1}, \bar M_{t-1}, \bar Z_{t-1}, \bar
A_{t-1}^{\d'})$. Intervened histories for other variables are defined analogously, as are counterfactual histories under the other intervention $\d^\star$, denoted with $H^\star$. Similarly, to ease notation, we let $A_t'=\d_t'(A_t, H_{A,t},\epsilon_t)$, with $A_t^\star$ defined similarly.

To identify the parameter, $\E[Y(\bar A^{\d'}, \bar J(\bar A^{\d^\star}))]$, the following assumptions, which combine elements of the identification assumptions from \cite{diaz2021nonparametricmtp} and \cite{diaz2023efficient}, are sufficient. Note that \ref{ass:iden} involves future histories rather than single timepoints only.

\begin{assumptioniden}[Conditional exchangeability of treatment and
  mediator]\label{ass:iden} Assume:
  \begin{enumerate}[label=(\roman*)]
  \item
    $U_{A,t}\indep (\underline U_{Z, t}, \underline U_{M, t},
    \underline U_{L, {t+1}}, \underline U_{A, t+1}) \mid H_{A,t}$ for all
    $t\in\{1,\ldots,\tau\}$; \label{ass:ceA}
  \item
    $U_{M,t}\indep (\underline U_{L, {t+1}},  \underline U_{A, t+1}, \underline U_{Z, {t+1}}) \mid H_{M,t}$ for all
    $t\in\{1,\ldots,\tau\}$. \label{ass:ceM}
  \end{enumerate}
\end{assumptioniden}

\begin{assumptioniden}[Positivity of treatment and mediator assignment]\label{ass:pos} Assume:
  \begin{enumerate}[label=(\roman*)]
  \item If $(a_t,h_{A,t})\in \supp\{A_t,H_{A,t}\}$ then
    $(\d^\star(a_t,h_{A,t}),h_{A,t})\in \supp\{A_t,H_{A,t}\}$ for
    $t\in\{1,\ldots,\tau\}$, and similarly for $\d'$; \label{ass:posA}
  \item If $h_{M,t}\in \supp\{H_{M,t}\}$ then
    $(m_t,h_{M,t})\in \supp\{M_t,H_{M,t}\}$ for
    $t\in\{1,\ldots,\tau\}$ and $m_t\in\supp\{M_t\}$. \label{ass:posM}
  \end{enumerate}
\end{assumptioniden}

If $A_t$ is not a discrete random variable, one further condition is required.

\begin{assumptioniden}[Piecewise-smooth invertibility of treatment modification]\label{ass:piece}

    If $A_t$ is not a discrete random variable, the conditional support of $A_t$ given $H_{A,t}=h_{A,t}$ admits a partition into subintervals where the restrictions of $\d'( \cdot, h_{a,t})  $ and $\d^\star( \cdot, h_{a,t})$ to each subinterval have differentiable inverse functions.  
\end{assumptioniden}

Our identification strategy makes use of a representation of longitudinal causal estimands by means of ``sequential regression" as proposed by \cite{Bang05} and further investigated by \cite{van2012targeted, luedtke2017sequential}, and \cite{rotnitzky2017multiply}. This approach considers a recursive regression of outcomes (or pseudo-outcomes) onto previous variables, starting at $t=\tau$ and continuing backward to $t=0$. As in \cite{diaz2023efficient}, we split the recursion over the two non-interventional variables $L$ and $Z$ as follows. We consider fixed, pre-specified interventions $\d'$, $\d^\star$, and $\bar m$. Starting with $\Q_{Z,\tau+1}= Y$, define
\begin{align}
  \Q_{L,t}(\bar h_{M,t},\ubar m_t)&= \E[\Q_{Z, t+1}(A{'}_{t+1}, H_{A, t+1}, \ubar
                                    m_{t+1})\mid M_{t} = m_t,
                                    H_{M,t} = h_{M,t}]\label{eq:qM}\\
  \Q_{Z,t}(a_t, \bar h_{A,t}, \ubar m_t) &= \E[\Q_{L,t}(\bar H_{M,t},\ubar m_t)\mid A_t =
                       a_t ,H_{A,t} = h_{A,t}].\label{eq:qZ}
\end{align}
Similarly, let $\Q_{M,\tau + 1} = 1$ and recursively define
\[\Q_{M, t}(a_t, \bar h_{A,t}, \ubar m_t)
  =\E[\one\{M_t = m_t\}\Q_{M, t+1}(A_t^\star, \bar H_{A,t+1}, \ubar
  m_{t+1})\mid A_t=a_t, H_{A,t}=h_{A,t}]. \]

When the meaning is clear, noting that the only non-random components of the functions $\Q$ are the user-specified $\d', \d^\star$, and $\bar m$, we may use the abbreviations $\Q_{L,t}'(\ubar m_t) = \Q_{L,t}(\bar H_{M,t},\ubar m_t)$, $\Q_{Z,t}'(\ubar m_t) =  \Q_{Z,t}(A^{'}_t, \bar H_{A,t}, \ubar m_t)$, and $\Q_{M, t}^\star(\ubar m_t) = \Q_{M, t}(A_t^\star, \bar H_{A,t}, \ubar
m_t)$. We give the following identification result, which is a generalization of Theorem 1 of \cite{diaz2023efficient}.

\begin{theorem}[Identification]\label{theo:iden}
  Under Assumptions~\ref{ass:iden} and~\ref{ass:pos}, the
  interventional parameter
  $\theta=\E[Y(\bar A^{\d'}, \bar J(\bar A^{\d^\star}))]$ is
  identified as
  $\theta = \sum_{\bar m} \varphi(\bar m)\lambda(\bar m)$, where
  $\varphi(\bar m)=\Q_{L,0}'(\bar m)$ and
  $\lambda(\bar m)=\Q_{M, 0}^\star(\bar m)$.
\end{theorem}

\begin{proof}
See Section \ref{sup:ident} of the supplement.
\end{proof}

\subsection{Inverse probability weighted identification}\label{sec:ipw-iden}

As is often the case in causal inference problems, it is also possible to identify effects of interest by using the treatment rather than the outcome mechanism. Let
$\g_t'(a_t \mid h_{A,t})$ and $\g_t^\star(a_t \mid h_{A,t})$ denote
the density functions of $A'_t=\d'(A_t, H_{A,t},\epsilon_t)$ and
$A^\star_t=\d^\star(A_t, H_{A,t},\epsilon_t)$, respectively,
conditional on $H_{A,t}=h_{A,t}$. Then we define
\[\G_{A,t}'(A_t, H_{A,t})=\frac{\g_t'(A_t \mid
    H_{A,t})}{\g_t(A_t \mid
    H_{A,t})},\,\,\G_{A,t}^\star(A_t, H_{A,t})=\frac{\g_t^\star(A_t \mid
    H_{A,t})}{\g_t(A_t\mid
    H_{A,t})},\,\,\,
  \G_{M,t}(H_{M,t},m_t)=\frac{\one\{M_t=m_t\}}{\g_{M,t}(M_t\mid
    H_{M,t})},\] as well as $\C_{l,u}' = \prod_{r=l}^{u}\G_{A,r}'$,
$\C_{l,u}^\star = \prod_{r=l}^{u}\G_{A,r}^\star$,
$\H_{l,u} = \prod_{r=l}^{u}\G_{M,r}$. (For notational brevity, we have omitted the
dependence on $A_t$, $H_{A,t}$, $H_{M,t}$, and $m_t$.) Then, the parameters from Theorem \ref{theo:iden} can also be identified as $\varphi(\bar m)=\E\{\C_{1,\tau}'\H_{1,\tau} Y\}$;
$\lambda(\bar m) = \E\{\C_{1,\tau}^\star \one\{\bar M = \bar m\}\}$; see Section \ref{sup:ipw} of the supplement for the derivation.

\subsection{Overview of estimation}

The foregoing analysis allows us to construct estimators of the parameter $\theta$ given in Theorem~\ref{theo:iden}. When using machine learning for estimation of the outcome and treatment mechanisms of the previous section, it is desirable to use doubly robust estimators, which guarantee efficiency and asymptotic normality under certain assumptions \citep{kennedy2024semiparametric}. For our problem, the construction of those estimators follows from a straightforward generalization of the methods of \cite{diaz2021nonparametricmtp} and \cite{diaz2023efficient}. These estimators leverage semiparametric theory to use both the outcome and treatment models to yield a combined estimator which is consistent if either one is correctly specified at each timepoint, and which can achieve $\sqrt{n}$-rate convergence even if the individual estimators converge at a slower rate. Please see Section \ref{ssec:estim} of the supplement for a complete description of the estimators.
Our estimation algorithm has been implemented in the open-source R package \textit{lcmmtp} \citep{lcmmtp2023}.

\section{Motivating application}\label{sec:aplica}

\subsection{Background}

Acute respiratory distress syndrome (ARDS) is a major cause of morbidity and mortality among COVID-19 patients. Treating ARDS often requires the use of respiratory support devices, ranging from nasal cannulae and face masks to more invasive methods such as mechanical ventilation via endotracheal tubes (``invasive mechanical ventilation" or IMV)~\citep{hasan2020mortality}. Given the physiological links between the lungs and kidneys, it has been suggested that mechanical ventilation may cause acute kidney injury/failure (AKI); possible mechanisms include oxygen toxicity and capillary endothelial damage leading to inflammation, hypotension, and sepsis~\citep{durdevic2020progressive}. 
AKI, which can render kidneys incapable of appropriately clearing critical toxins from the body or maintaining appropriate blood and interstitial volume, complicates about 30\% of ICU admissions and increases the risk of in-hospital mortality and long-term morbidity and
mortality~\citep{kes2010acute}. 

In the early weeks of the COVID-19 pandemic, recommendations from international health organizations advocated for early intubation of patients in an effort to safeguard healthcare workers from contracting infection and to reduce complications resulting from late intubations when the patient is unresponsive or apneic~\citep{papoutsi2021effect}. Over time, however, as clinicians became more familiar with the progression of the disease, guidance evolved toward postponing intubation, partially driven by reports showing high mortality rates for mechanically ventilated patients, potentially attributable to heightened risk of secondary infections, lung injury, and damage to other organs including the kidneys, due to ventilation~\citep{tobin2006principles, bavishi2021timing}.

We illustrate our methods in this application with the goal of providing insight into effects operating through the pathway IMV$\to$AKI$\to$Death and effects operating directly through IMV$\to$Death. We answer the following question, which has been cited by an expert panel on lung-kidney interactions in critically ill patients as an area that demands further research \citep{joannidis2020lung}:
\textit{What is the effect of invasive mechanical ventilation on
death among COVID-19 patients, and how much of it operates through causing acute kidney injury?}

We use a dataset consisting of approximately 3,300 patients who did not have a previous history of chronic kidney disease (CKD) and who were hospitalized with COVID-19 at the NewYork-Presbyterian Cornell, Queens, or Lower Manhattan hospitals between March 3\textsuperscript{rd} and May 15\textsuperscript{th}, 2020. The analytical dataset was created
in a two-step approach. 
First, data pertaining to demographics, comorbidities, intubation, mortality, and discharge were gathered from electronic health records chart reviews and stored in a secure REDCap database~\citep{goyal2020clinical}.
To this data were added items from the Weill Cornell Critical carE Database for Advanced Research (WC-CEDAR), 
a data repository containing patient information collected over the course of care, including procedures, diagnoses, medications, and laboratory data~\citep{schenck2021critical}. Note that while all patients were hospitalized with COVID-19, they did not all necessarily suffer from ARDS specifically.

This data was previously described and studied in \cite{diaz2022causalcomp} and \cite{hoffman2024studyingcontinuoustimevaryingandor}. The latter considered only effects on mortality without regard to AKI. The former treated AKI and other-cause mortality as competing risks; we take a different perspective by considering AKI as a mediator in the pathway from IMV to mortality. This orientation presents a more intricate view of the causal connections among these variables. Essentially, AKI and mortality are not only competing risks; they can jointly occur, with AKI potentially influencing mortality. By examining AKI's mediating role, we aim to unravel the indirect ways in which IMV may contribute to mortality via AKI. Should this causal pathway exist, it could provide novel opportunities for medical interventions aimed at preventing or treating AKI, subsequently reducing the mortality rate. (As mentioned in Section \ref{sec:intro}, these options include the modulation of intravascular volume status \citep{grams2011fluid, glassford2011fluid} and various pharmacological interventions \citep{pickkers2022new}.) Further, we analyze how various baseline variables are associated with the estimated direct and indirect effects and discuss how such considerations might inform treatment decisions regarding IMV.

\subsection{Time-varying variables, mediator, and modified treatment policy}

This study uses data at the daily level, and study time begins on the day of hospitalization. The treatment variable is categorized into three levels: no supplemental oxygen, supplemental oxygen not including IMV, and IMV. The main goal is to estimate the overall causal
effect of IMV on mortality/survival rates and to decompose it into effects operating through AKI and effects operating independently of it. To measure the total effect of IMV on survival, we posit an intervention that would assign a noninvasive type of oxygen support to patients requiring oxygen support. This relatively modest intervention (as opposed to one that requires everyone to receive the same level of oxygen support) is motivated by a desire to preserve positivity  (\ref{ass:pos},\ref{ass:posA}) and avoid extreme inverse propensity score weights.
In actuality, 381 of the patients received intubation during the time period considered.

Formally, for $A_t\in\{0,1,2\}$ corresponding to the end of day $t$,  $0$ indicates no supplemental
oxygen, $1$ indicates supplemental oxygen not
including IMV, and $2$ indicates IMV. We consider the following
intervention:
\begin{equation}\label{eq:exmtp}
  \d_t(a_t) =
  \begin{cases}
    1 &\text{ if } a_t=2\\%
    a_t & \text{ otherwise.}
  \end{cases}
\end{equation}%
In words, this means that the intervention leaves the observed treatment value unchanged, except for days when the patient was intubated. For those days, the intervention changes the treatment to supplemental oxygen \textit{without} IMV. 

In constructing the analytical dataset, it is important that the variables preserve the time-ordering of the DAG in Figure~\ref{fig:dag}. Therefore, we constructed the dataset using the following procedure. First, we categorize patients into three groups: those whose first event was an intubation, those whose first event was AKI, and those who had neither event. 
For patients in the first group, we anchor the creation of their record at the time at which the intubation event occurred, denoted $t_I$. %
The variables $L_{t_{I-1}}$ and $Z_{t_{I-1}}$ are then recorded as those observed in approximately 24-hour windows before $t_I$, where $L_{t_{I-1}}$ contains variables measured in the first half of this approximate 24-hour interval, and $Z_{t_{I-1}}$ contains variables measured in the second half of this approximate 24-hour window. The time interval is approximate because the difference in hours between hospital admission and intubation is not necessarily a multiple of 24. The mediator $M_{t_{I-1}}$ is then set as $0$, and the exposure $A_{t_{I-1}}$ is recorded as the respiratory support in the previous 24-hour window. This process is iterated until we have data for all the timepoints prior to $t_I$. For times greater than or equal to $t_I$, we follow a similar approach, but in addition to capturing $Z_t$, $L_t$, and $A_t$, we measure whether the patient was diagnosed with AKI at every time window. AKI was defined using creatinine values in accordance with the ``Kidney Disease: Improving Global Outcomes" (KDIGO) definition \citep{khwaja2012kdigo}. Either of the following criteria was required: (a) serum creatinine change of greater than or equal to 0.3 mg/dL within 48 h, or (b) serum creatinine greater than or equal to 1.5 times the baseline serum creatinine known or assumed to have occurred within the past 7 days.

If and when a patient is diagnosed with AKI (call this time $t_{M}$), we begin a new 24-hour interval, i.e., $L_{t_M}$, to ensure that the temporality of variables measured after $t_M$ is maintained. A similar procedure was used for patients whose first event was the diagnosis of AKI. Data for patients in the third group, who had neither event, was simply divided into approximate 24-hour time intervals using the anchor times of hospital admission and censoring or death date. During this process, we ensured that measurement of $L_t$ precedes $Z_t$ within each time interval. A GitHub repository of analysis code is available at \url{https://github.com/CI-NYC/lcmmtp-application}. For more details, including the list of measured covariates, see Section \ref{ssec:data} of the supplement.

\subsection{Results}\label{sec:res}

We used our proposed methods to estimate the total, direct, and indirect effects on mortality of invasive mechanical ventilation through acute kidney injury. All regressions were estimated using \textit{lcmmtp} \citep{lcmmtp2023} and Super Learner \citep{van2007super}, building a convex combination of three predictors: generalized linear models with $\ell_1$ penalty, multivariate adaptive regression splines, and extreme gradient boosting of regression trees. To speed up computations, we used three folds for cross-fitting and used four-day survival as our outcome of interest. Some observations had large weights $\C_{t,s}$ and $\H_{t,s}$, so we truncated the corresponding estimates at their 99th percentiles to avoid empirical violations of the positivity assumption. The results of the main analysis are presented in Table~\ref{tab:total}, where we contrast the mortality under an intervention to prevent censoring versus the mortality under an intervention to prevent censoring and to %
prevent intubation as discussed above. Specifically, we estimate the total effect defined as $\E[Y(\bar A^\d, \bar J(\bar A^\d))-Y(\bar A, \bar J(\bar A))]$, where $\d$ is defined in (\ref{eq:exmtp}) and $Y$ is four-day \textit{survival rate} under an intervention to prevent censoring. The indirect effect is thus defined as $\E[Y(\bar A, \bar J(\bar A^\d))-Y(\bar A, \bar J(\bar A)) ]$, and the direct effect is $\E[ Y(\bar A^\d, \bar J(\bar A^\d))-Y(\bar A, \bar J(\bar A^\d)) ]$.

Although there is high uncertainty in the results, the point estimates support the hypothesis that invasive mechanical ventilation reduces mortality through mechanisms other than acute kidney injury. However, the effect through acute kidney injury is larger than the effect through all other mechanisms, resulting in an overall harmful effect of invasive mechanical ventilation on mortality. 
The estimated value of the total effect indicates that %
preventing intubation
 would reduce mortality by 5.9 percentage points, with a standard error of 3.8 percentage points.

For an additional analysis of effect modification by observed covariates, see Section \ref{ssec:modify} of the supplement.

\subsection{Mediation and treatment decisions}\label{sec:decisions}

Here, we present a conjectured approach for how mediation analysis might inform treatment decisions. 
For the purpose of exposition, consider a setting with a single timepoint where we label the observed data as $O=(W, A, M, Y)$, where $A$ is intubation status, $M$ is the occurrence of AKI, $W$ are a set of pre-treatment confounders, and $Y$ is survival. In studies of medical interventions, it is often of interest to decide an optimal treatment rule; that is, a function $d(W)$ mapping covariates to treatment values that maximize the probability of survival of a patient presenting with measured variables $W$. However, in the case of critical care for a relatively new and evolving disease, doctors would likely prefer to incorporate contextual indicators, patient-specific factors, and clinical judgment rather than to strictly rely on a mathematical rule constructed from previous data. (See \cite{gallifant2022artificial} for a critical review of existing academic literature promoting the use of artificial intelligence for mechanical ventilation decisions and \cite{mathur2019artificial} for an overview of professional challenges that ICU physicians might face in applying algorithmic systems to their practice.) On the other hand, a mediation analysis of the kind we are proposing may be useful in providing a \textit{partial} treatment rule, insofar as it might identify a subset of patients who should \textit{not} receive intubation (i.e., those whose respiratory failure might recover without IMV and for whom IMV would introduce unnecessary risks), while leaving undecided the treatment given to patients outside this subset.

Specifically, we note again that the overall effect of intubation on survival can be represented as a combination of the effects through the mediating variable (AKI), and ``direct" effects that do not operate through the mediator. (Because we are estimating ``interventional" rather than ``natural" effects \citep{vanderweele2014effect, miles2015partial}, this ``overall effect" is not quite equal to the ordinary average treatment effect of intubation on survival, but it can be thought of in a similar way as a heuristic approximation; as described in Section \ref{sec:intro}, this substitution is done for the purpose of identification.)
In principle, both the indirect and direct effects could be beneficial (increasing probability of survival) or harmful (decreasing probability of survival). When these effects are in contrary directions, it may be said that the process exhibits ``inconsistent mediation" \citep{mackinnon2000equivalence}. In the current clinical context, we might discount the existence of beneficial indirect effects (i.e., that intubation causes survival by preventing AKI), even at the individual level. Then, for any individual predicted to experience a harmful \textit{direct} effect of treatment, the total effect of intubation on survival can be expected to be at least as poor. This line of reasoning provides a conservative basis for \textit{excluding} a subset of patients from treatment. 

By contrast, one could attempt to estimate the total effect for all patients and exclude any patients with a negative value, but this precludes the possibility of additional decision-making input from physicians and ignores the possibility of ameliorating the indirect effects (i.e., that additional interventions could be put in place to mitigate the probability of AKI for patients under intubation). 

While we are not advocating that such a treatment rule be applied immediately to patients similar to the ones under study, we believe the information gleaned from this analysis could be useful in future research. In particular, whether beneficial indirect effects could be plausible for some individuals, and what interventions could be put in place to prevent AKI in intubated individuals, do not appear to be well-established in the literature, but can be informed from such analyses. Aside from the current application, this framework could be useful for any treatment scenario where such conditions (the ability to mitigate harmful indirect effects and/or the absence of beneficial indirect effects) are fulfilled.

\section{Discussion}\label{sec:discussion}

Causal mediation analysis is a complex problem, especially in settings with time-varying and/or continuous exposures, mediators, and confounders; we have provided a framework and technique to handle these intricacies.
We have argued that mediation analysis can be helpful in clinical decision-making and in formulating preliminary models of causal mechanisms. In addition, we have highlighted the medical importance of inconsistent mediation and its potential role in informing critical care. For this purpose, estimating conditional effects can help to identify the relevant subgroups and to understand the biological factors driving the inconsistency.
The application to COVID-19 hospital data, while leaving much room for uncertainty, lends some credence to concerns regarding risks of early intubation through acute kidney injury. Further research might investigate approximate methods to ameliorate the computational costs associated with the sequential regression, which might grow quickly with the number of timepoints.

\begin{table}[H]
	
	\caption{Total, direct, and indirect effect of delaying invasive mechanical ventilation by one day on survival operating through acute kidney failure (presented as intervention minus baseline; positive values indicate that delaying intubation is beneficial).}
	
	\centering
	
	\begin{tabular}{rrr}
		
		\hline
		
		&Effect & Standard Error \\ 
		
		\hline
		
		Total    & 0.059 & 0.038 \\ 
		
		Direct (i.e., not through AKI)  & -0.024 & 0.039 \\ 
		
		Indirect (i.e., through AKI) &  0.083 & 0.027 \\ 
		
		\hline\label{tab:total}
		
	\end{tabular}
	
\end{table}

\section{Acknowledgments}

This work was supported through a Patient-Centered Outcomes Research Institute (PCORI) Project Program Award (ME-2021C2-23636-IC).
\clearpage
\appendix
\section*{Supplementary Materials}

\setcounter{subsection}{0} %
\renewcommand{\thesubsection}{S\arabic{subsection}} %
\setcounter{equation}{0}
\setcounter{theorem}{0}
\setcounter{section}{0}

\renewcommand{\thesection}{S\arabic{section}}
\renewcommand{\thetable}{S\arabic{table}}
\renewcommand{\thefigure}{S\arabic{figure}}
\renewcommand{\theequation}{S\arabic{equation}}

\section{Technical note on notation}

In the proofs that follow, any variable with index $t \leq 0$ should be interpreted as null, and a distribution conditional on null variables is marginal. Any summation expression with no terms (e.g., $\sum_{t+1}^{t} W_t$) is equal to zero, and any product with no factors is equal to one.

\section{Identification (Theorem \ref{theo:iden})}\label{sup:ident}

\begin{proof}
    By the law of total expectation, we can write \[\theta=\E[Y(\bar A^{\d'}, \bar J(\bar A^{\d^\star}))] \]
    \[
    =\E[\E[Y(\bar A^{\d'}, \bar J(\bar A^{\d^\star})) \mid J(\bar A^{\d^\star})]]\]
   \[ = \sum_{\bar m} \E[Y(\bar A^{\d'}, \bar J(\bar A^{\d^\star})) \mid \bar J(\bar A^{\d^\star}) = \bar m] \cdot \P(\bar J(\bar A^{\d^\star}) = \bar m)\]
   \[ = \sum_{\bar m} \E[Y(\bar A^{\d'}, \bar m )\mid \bar J(\bar A^{\d^\star}) = \bar m] \cdot \P(\bar J(\bar A^{\d^\star}) = \bar m)\]

   Since $\bar J$ is defined to be distributed like $\bar M$ and independent of all data,

   \[
   = \sum_{\bar m} \E[Y(\bar A^{\d'}, \bar m )] \cdot \P(\bar M(\bar A^{\d^\star}) = \bar m)
   \]

   Given policies  $\d'$ and $\d^\star$, we let $\varphi(\bar m) = \E[Y(\bar A^{\d'}, \bar m )]$ and $\lambda( \bar m) = \P(\bar M(\bar A^{\d^\star}) = \bar m)$ and will show that $\varphi$ and $\lambda$ are equivalent to the expressions given in the theorem.

   First, for $\varphi$, we consider the relabeled dataset where for all $t\leq \tau$, $(L_t, Z_t) = (C_{2t-1}, C_{2t})$ and $(A_t, M_t) = (E_{2t-1}, E_{2t})$. In this case, the expression $\E[Y(\bar A^{\d'}, \bar m )]$ is the expected value of the counterfactual outcome of the longitudinal modified treatment policy setting the values of $E$ to the corresponding values $\bar A^{\d'}$ and $\bar m$. Similarly to \cite{diaz2023efficient}, this expression is identified by the following strategy.

   Let $\mu_{2\tau+1} = Y$ and for $t = 2\tau, \dots, 1$ recursively define 
   \[
   \mu_{t}(e_t, h_t) = \E[\mu_{t+1}(E'_{t+1}, H_{t+1}) \mid E_t = e_t, H_t = h_t]
   \]
   where $H_t$ denotes all variables prior to $E_t$. Then by integration, $\E[Y(E^{\d'})] = \E[\mu_1(E_1', C_1)]$.

   Since the desired intervention on $\bar m$ is not stochastic (and thus the mediation component of the intervention $\d'$ can be passed into the $\Q$-functions as the future histories $\ubar m_t$), it follows that substituting variables is sufficient to recognize, when $\d'$ is the intervention setting $E$ to the values of $\bar A^{\d'}$ and $\bar m$, the relation $\E[\mu_1(E_1', C_1)] = \Q_{L,0}'(\bar m)$, as in the Theorem.

    For $\lambda(\bar m) = \P(\bar M(\bar A^{\d^\star}) = \bar m)$, we again relabel the data, this time grouping $L, Z, M$ as $S_t = (L_t, Z_{t-1}, \bar M_{t-1})$ (and again suitably altering the meaning of $H_t$ to denote all variables before $A_t$, or $H_{A,t}$ in the original labeling). Then $\bar M(\bar A^{\d^\star})$ is simply a projection of the outcome $S_{\tau+1}$ under the longitudinal modified treatment policy $\bar A^{\d^\star}$. Denoting this projection as $S_t^m = \bar M_{t-1}$, we can write $\P(\bar M(\bar A^{\d^\star}) = \bar m) = \E[\1(S_{\tau+1}^m = \bar m)]$, and we can again apply a similar identification strategy.

    As a shorthand, let $\1_{\tau+1}^{\bar m} =\1(S_{\tau+1}^m = \bar m)$. Let $\rho_{\tau+1} = \1_{\tau+1}^{\bar m}$, and for $t = \tau, \dots, 1$ recursively define 
    
   \[
   \rho_{t}(a_t, h_t) = \E[\rho_{t+1}(A^{\star}_{t+1}, H_{t+1}) \mid A_t = a_t, H_t = h_t]
   \] 

    Then by integration we have $\lambda(\bar m) = \E[\rho_1(A_1^{\star}, S_1)]$. We will show that this expression is equivalent to $\Q_{M,0}^\star(\bar m)$ as in the Theorem.

    First, by definition we have
    \[
    \rho_{\tau}(a_\tau, h_\tau) = \E[\1_{\tau+1}^{\bar m} \mid A_\tau = a_\tau, H_t = h_\tau]
    \]

The conditioning event fixes all mediator values but the final one, so we can write
    
    \[
    = \1(S_\tau^m = \bar m_{\tau-1}) \P( M_{\tau} = m_\tau \mid  A_\tau = a_\tau, H_t = h_\tau)
    \]
    \[
    = \1(S_\tau^m = \bar m_{\tau-1}) \E[\1( M_{\tau} = m_\tau) \mid  A_\tau = a_\tau, H_t = h_\tau]
    \]

    By definition,

    \[
   = \1(S_\tau^m = \bar m_{\tau-1}) \Q_{M,\tau}^\star( \ubar m_\tau)
    \]

    Then for $t=\tau-1$, we similarly have
    \[
    \rho_{\tau-1}(a_{\tau-1}, h_{\tau-1}) =  \E[\rho_\tau(A_\tau^{\star}, H_{\tau}) \mid A_{\tau-1} = a_{\tau-1}, H_t = h_{\tau-1}]
    \]
    \[
    = \E[\1(S_\tau^m = \bar m_{\tau-1})  \Q_{M,\tau}^\star( \ubar m_\tau) \mid A_{\tau-1} = a_{\tau-1}, H_t = h_{\tau-1}]
    \]
    \[
    = \1(S_{\tau-1}^{m} = \bar m_{\tau-2}) \E[ \1 (M_{\tau-1}=m_{\tau-1})  \Q_{M,\tau}^\star( \ubar m_\tau) \mid A_{\tau-1} = a_{\tau-1}, H_t = h_{\tau-1}]
    \]
    \[
    =\1(S_{\tau-1}^{m} = \bar m_{\tau-2}) \Q_{M,\tau-1}^{\star}( \ubar m_{\tau-1})
    \]
    
    Repeated iterations yield the pattern
    \[
    \rho_t(a_t, h_t) = \1(S_t^{m} = \bar m_{t-1}) \Q_{M,t}^\star(\ubar m_{t})
    \]
    Thus
    $\lambda(\bar m)=\E[\rho_1(A_1^{\star}, S_1)] = \E[\Q_{M,1}^\star(\bar m)] = \Q_{M,0}^\star( \bar m)$ as desired.
\end{proof}

\section{IPW identification}\label{sup:ipw}

\cite{diaz2021nonparametricmtp} gives, for a general longitudinal modified treatment policy $\bar A^\d$ and outcome $Y$ (without explicit mediation or intermediate confounders):

\[
\E[Y(\bar A^\d)] = \E\left[Y \prod_{t=1}^\tau \frac{g_t^\d(a_t|h_{A,t})}{g_t(a_t|h_{A,t})}\right]
\]

Applying this result after relabeling the data exactly as in the preceding Section \ref{sup:ident} (respectively for each estimand) immediately yields the given IPW expressions for $\varphi$ and $\lambda$.

\section{More on interpretation of interventional effects}\label{ssec:interp}
\subsection{Difference between interventional and natural effects}
Natural direct and indirect effects aim to decompose the total effect of an exposure \( A \) on an outcome \( Y \) through a mediator \( M \). However, when there exists an intermediate confounder \( Z \) that is affected by \( A \) and also influences both \( M \) and \( Y \), identifying these natural effects becomes mathematically infeasible.

Consider the following structural causal model (SCM), represented by the structural equations:
\begin{align*}
L &= f_L(U_L), \\
A &= f_A(L, U_A), \\
Z &= f_Z(L, A, U_Z), \\
M &= f_M(L, A, Z, U_M), \\
Y &= f_Y(L, A, Z, M, U_Y),
\end{align*}
where \( L \) represents baseline covariates, and \( U_L, U_A, U_Z, U_M, U_Y \) are independent exogenous variables. The variable \( Z \) is an intermediate confounder affected by \( A \), which also influences both \( M \) and \( Y \). This model represents a single timepoint, but similar principles apply to the longitudinal case.

The natural direct effect (NDE) and natural indirect effect (NIE) are defined as follows:

\[
\text{NDE}(a, a') = \E[Y(a, M(a'))] - \E[Y(a', M(a'))],
\]

\[
\text{NIE}(a, a') = \E[Y(a, M(a))] - \E[Y(a, M(a'))],
\]
where \( Y(a, M(a')) \) denotes the potential outcome when the exposure is set to \( a \) and the mediator is set to the value it would have taken under exposure \( a' \).

To compute \( Y(a, M(a')) \), we need to consider the outcome \( Y \) when the exposure is set to \( a \) and the mediator is set to \( M(a') \). However, \( Y(a, M(a')) \) inherently depends on both \( Z(a) \) and \( Z(a') \) because:

\[
Y(a, M(a')) = f_Y\left(f_L(U_L), a, Z(a), M(a'), U_Y\right),
\]
and
\[
M(a') = f_M\left(f_L(U_L), a', Z(a'), U_M\right),
\]
where
\[
Z(a) = f_Z\left(f_L(U_L), a, U_Z\right), \quad Z(a') = f_Z\left(f_L(U_L), a', U_Z\right).
\]

The problem arises because \( Z(a) \) and \( Z(a') \) are potential outcomes of \( Z \) under different exposures \( a \) and \( a' \), respectively, and they are functions of the same exogenous variable \( U_Z \). They are not independent, nor can the joint distribution of \( (Z(a), Z(a')) \) be identified.

Interventional direct and indirect effects offer an alternative framework that overcomes this limitation. The interventional direct effect (IDE) is defined as the difference in expected outcomes when the exposure is set to \( a \) and the mediator is intervened on to follow the \textit{distribution} it would have under exposure \( a' \), compared to when both the exposure and mediator are set to \( a' \). Mathematically, this can be expressed as:
\[
\text{IDE}(a, a') = \E[Y(a, J(a'))] - \E[Y(a', J(a'))],
\]
where \( J(a') \) represents the mediator values drawn from the distribution under exposure \( a' \).

Similarly, the interventional indirect effect (IIE) measures the difference in expected outcomes when the exposure is set to \( a \) with the mediator drawn from its distribution under exposure \( a \), versus exposure \( a \) with the mediator drawn from its distribution under \( a' \):
\[
\text{IIE}(a, a') = \E[Y(a, J(a))] - \E[Y(a, J(a'))].
\]

Defined this way, the effects \textit{are} identifiable, because $J(a)$ is defined only to come from the distribution of $M(a)$ rather than as the value of $M$ for a given individual (which would depend on that individual's value for $Z(a)$). In particular, \( \E[Y(a, J(a'))] \) does \textit{not} depend on the joint distribution of \( (Z(a), Z(a')) \) because $J(a')$ is by definition independent of other variables.

While this interventional parameter is more plausibly identifiable, it has some drawbacks in terms of interpretability. For example, it does not decompose the average treatment effect:

\[
\text{IDE}(a, a') + \text{IIE}(a, a') = \E[Y(a, J(a))] - \E[Y(a', J(a'))]
\]
which is \textit{not} necessarily equal to $\E[Y(a)] - \E[Y(a')]$. This is in contrast to the natural effects, since $\E[Y(a, M(a))] - \E[Y(a', M(a'))]$ \textit{does} equal  $\E[Y(a)] - \E[Y(a')]$ since $Y(a, M(a)) = Y(a)$.

Further, \cite{miles2023causal} has shown that without additional assumptions, the interventional effects do not necessarily equal zero when there are no individual-level effects, which also detracts from its scientific interpretability.

\subsection{Subtleties for longitudinal data}

When treatment is applied via an LMTP at multiple timepoints, the intuitive distinction between direct and indirect effects becomes somewhat more complicated. However, it can be made more concrete by recognizing the fact presented in Proposition \ref{prop:decomp} which follows below.

\begin{definition}
For a given causal DAG, we say a ``causal pathway"  from some node $A$ to some node $Y$ is a sequence of nodes $(A, V_1, V_2, \dots, V_n, Y)$ connected by arrows where a manipulation of the value of $A$ has the effect of changing the value of $Y$ by directly changing the value of $V_1$, which in turn changes the value of $V_2$, etc., up to $Y$.
\end{definition}

\begin{definition}
    A causal pathway $(A, V_1, \dots, V_n, Y)$ is ``blocked" for some comparison between interventions $\alpha$ and $\beta$ if $\alpha$ and $\beta$ have the effect that at least one of $A, V_1, V_2, \dots, V_n, Y$ has the same value under $\alpha$ and $\beta$ for each observation.
\end{definition}

For a more in-depth description of the concept of causal pathways, the interested reader may consult, for example, \cite{pearl1998graphs}.

\begin{proposition}\label{prop:decomp}
    In the structural causal model given in Figure \ref{fig:dag} of the main paper, let the nodes labeled $A$ and $M$ be denoted the ``interventional nodes." If every causal pathway from $A_i$ (for some $i$) to $Y$ has $M_j$ (for some $j$) as the last interventional node, then the natural direct effect of any LMTP is zero. 
\end{proposition}

\begin{proof}
    Consider arbitrary LMTPs $\d_1$ and $\d_2$. Then the natural direct effect is composed of differences of the form 
    \[
    Y(\bar A^{\d_2}, \bar M (\bar A^{\d_1})) -  Y(\bar A^{\d_1}, \bar M (\bar A^{\d_1}))
    \]\\
    The two terms of this contrast arise from two causal models where the mediator variables are identical. Therefore, any pathway $A_i \rightarrow M_j \rightarrow \mathbf V \rightarrow Y$ is blocked and cannot give rise to differences in $Y$, where $\mathbf V$ is any set of nodes that does not include a treatment node. More precisely, any node after $M_j$ which is not a treatment node is, for both terms of the contrast, generated from a causal model with $M_j$ arising from $\d^1$, and therefore there can be no differences through such a path.
\end{proof}

This means, in other words, that effects that have a mediator variable as the \textit{proximate} interventional cause of the outcome are not captured in the natural direct effect. 

However, with LMTPs, there are some subtleties regarding causal pathways that do not arise with static interventions. Importantly, it is not the case that \textit{any} pathway through a mediator is not captured in the natural direct effect, because the mediator can affect later treatment nodes through the LMTP, as the following example shows.  

Consider variables $(A_1, M_1, A_2, M_2, Y)$, where $A$ are treatments, $M$ are mediators, and $Y$ is the outcome. We can consider the effect of setting $A_1$ to $1$ and letting $A_2$ arise naturally, versus allowing both treatments to arise naturally. Then the natural direct effect is, explicitly, 
\[
\E[Y(A_1 = 1, M_1(A_1), A_2(A_1 = 1, M_1(A_1 = 1)), M_2(A_1, M_1(A_1), A_2(A_1, M_1(A_1)))) - Y]
\]
which we can abbreviate for clarity as 
\[
NDE = \E[Y(A_1 = 1, M_1, A_2(A_1 = 1, M_1(A_1 = 1)), M_2) - Y]
\]
where $M_1, M_2$ are the ``natural" values of the mediator. The presence of the counterfactual expression $A_2(A_1 = 1, M_1(A_1 = 1))$ indicates that if there is a causal pathway $A_1 \rightarrow M_1 \rightarrow A_2 \rightarrow Y$, then it will appear as part of the direct effect even though it involves a mediator. 

A few more remarks are in order about this issue. First, the interventional effects we consider (as opposed to natural effects) need not respect a similar restriction to Proposition \ref{prop:decomp}. The analogous limitation is true in the single-timepoint case as well, even for static treatments \citep{miles2023causal}. Further, we have stated only a condition about when the direct effect must be zero, rather than attempting to quantify its size when it is nonzero. We conjecture that its size will correspond to the aggregation of pathways with a treatment node as the last interventional node in the pathway (and, correspondingly, that the indirect effect will correspond to the aggregation of pathways with a mediator node as the last), but precise results seem challenging for nonparametric models and require further research.

\section{More on estimation}\label{ssec:estim}
The theory of semiparametric inference allows for the analysis of estimators of  parameters (e.g., causal effects) without restricting $\P$ to a parametric model. Standard references include \cite{vanderVaart98} and \cite{Bickel97}. In our case, we are particularly interested in \textit{nonparametric} inference which places \textit{no} restrictions on $\P$ other than the identification assumptions, particularly those related to positivity.

A given estimator $\hat \theta$ of $\theta$ is consistent and asymptotically linear if for some function $\S$ (which may depend on nuisance parameters) of the observations $O$ we have

\[
\sqrt{n}(\hat \theta - \theta) = \frac{1}{\sqrt{n}} \sum_{i=1}^n \S(O_i) + o_P(1)
\]

By the central limit theorem, such an estimator converges to a normal random variable with variance $\var[\S(O)]$ at the rate of $\sqrt{n}$. The function $\S$ is called the $\textit{influence function}$ of the estimator, and this convergence allows the construction of Wald-type confidence intervals. A central object in semiparametric theory is the \textit{efficient influence function}, the variance of which provides an asymptotic lower bound on the variance of consistent regular and asymptotically linear estimators of a given parameter. Thus, if one can prove that a proposed estimator is asymptotically linear and has variance equal to the efficient influence function (this occurs when $\S$ is the efficient influence function), then that estimator is a strong candidate for use in data analyses.

A considerable body of literature has studied semiparametric estimation of total causal effects of static and dynamic interventions using sequential regression \citep[e.g.,][]{Robins00,vanderLaan2003, Bang05,
  vdl2006targeted, vanderLaanRose11,
  luedtke2017sequential,rotnitzky2017multiply, vanderLaanRose18}. Typically, the approach involves the estimation of both the regression functions (in this case, the prediction of $Z, M,$ and $L$ conditional on their past) and the probability of treatment. These procedures have the property that the rate of convergence of the estimator for the treatment effect is equal to the \textit{product} of the rates of convergence of the outcome and treatment estimators; this allows flexible machine learning models that converge at rates slower than ``parametric" $\sqrt{n}$ rate, as long as their product converges at $\sqrt{n}$ rate (e.g., each nuisance parameter can be estimated at $n^{1/4}$ rate). In particular, if either model is consistent, then the treatment effect estimator is consistent.

The requisite $n^{1/4}$-rate estimator convergence is achievable for some certain classes of flexible predictors, including LASSO \citep{bickel2009simultaneous}, highly adaptive LASSO \citep{benkeser2016highly}, random forests \citep{wager2015adaptive}, and neural networks \citep{chen1999improved}, though we have not verified this fact for each estimator implemented in Super Learner. One may also consult the discussion of convergence rates for machine learning estimators in Section 3 of \cite{chernozhukov2018double}. We also note that if (correctly specified) parametric models are used, then the assumption can be easily verified. Overall, while the theoretical convergence rates of nonparametric estimators may be difficult to verify, we recommend the doubly robust procedure we have presented because it provides the best chance for appropriate uncertainty quantification while minimizing the possibility of error due to model misspecification.

As given below, the algebraic form of the efficient influence function for our parameter is identical to that of the parameter studied in \cite{diaz2023efficient}, except that the definitions of the IPW functions ($\G'$ and $\G^\star$) and the iterative regressions ($\Q$-functions)  change to accommodate the LMTP generalization. The proofs and arguments for all the following claims mirror those of \cite{diaz2023efficient} and are therefore omitted here. To give an expression for the efficient influence function for $\theta$, we must introduce the following functions, which can be understood intuitively as doubly robust analogs of the $\Q$- functions (see \cite{rubin2007doubly}).
\begin{equation}\label{eq:dl}
\D_{L,t}(\ubar X_t, \ubar m_t) = \sum_{s=t}^\tau \C^{'}_{t+1,s} \H_{t,s}\{\Q_{Z, s+1}-\Q_{L,s}\} + \sum_{s=t+1}^\tau \C^{'}_{t+1,s} \H_{t,s-1}\{\Q_{L, s}-\Q_{Z,s}\} + \Q_{L,t} 
\end{equation}
\begin{equation}\label{eq:dz}
\D_{Z,t}(\ubar X_t, \ubar m_t) = \sum_{s=t}^\tau \C^{'}_{t,s} \H_{t,s}\{\Q_{Z, s+1}-\Q_{L,s}\} + \sum_{s=t}^\tau \C^{'}_{t,s} \H_{t,s-1}\{\Q_{L, s}-\Q_{Z,s}\} + \Q_{Z,t}
\end{equation}
\begin{equation}\label{eq:dm}
\D_{M,t}(\ubar X_t, \ubar m_t) = \sum_{s=t}^\tau \C^{*}_{t,s} [\prod_{k=t}^{s-1} \1(M_k = m_k)] [\1(M_s = m_s) \Q_{M, s+1} - \Q_{M,s}] + \Q_{M,t}
\end{equation}
Then, we have the following:

\begin{theorem}[Efficient influence function for
  $\theta$]\label{theo:eif}
  The efficient influence function for $\theta$ in the nonparametric model is given by
  \[\S(X,\eta)= \sum_{\bar m\in \bar{\mathcal M}}\left[\{\D_{Z,1}(X,\bar m;\eta)
    - \varphi(\bar m)\}\lambda(\bar m) + \{\D_{M,1}(X,\bar m;\eta)
    - \lambda(\bar m)\}\varphi(\bar m)\right]\]

    where $\bar{\mathcal{M}}$ is the space of possible values for the mediator variables.
\end{theorem}

\begin{proof}
 See the proof of Theorem 2 of \cite{diaz2023efficient}.
\end{proof}

The algorithm used to estimate $\theta$ is given as Algorithm 1 in \cite{diaz2023efficient}, where in our case the IPW weights are replaced with probability ratios as described in Section \ref{sec:ipw-iden}. The algorithm consists of the following steps. First, we set $\hat \D_{Z,\tau+1}=Y$ and $\hat \D_{M, \tau+1}=1$. Then, beginning at $t=\tau$, we calculate $\hat \Q_{L,t}$ as the nonparametric regression prediction of $\hat \D_{Z,t+1}$, onto preceding variables. (The method is agnostic as to the choice of regression technique, but we recommend nonparametric or ``data-adaptive" approaches; see Section \ref{sec:res} of the main text for the specific tools we use in our data analysis.) Then, we nonparametrically estimate the functions $\g$ and $\G$ (as in Section \ref{sec:ipw-iden}), and $\hat \D_{L,t}$ is calculated by plugging in all the nuisance functions estimated so far into its definition (in equation \ref{eq:dl}). Similarly, $\hat \Q_{Z,t}$ comes from a nonparametric regression of $\hat \D_{L,t}$, and $\hat \Q_{M,t}$ comes from a nonparametric regression of $\1(M_t=m_t) \hat \D_{M,t+1}$. Then $\hat \D_{Z,t}$ is estimated from equation \ref{eq:dz} and $\hat \D_{M,t}$ is estimated from equation \ref{eq:dm}. This is repeated for $t=\tau-1, \tau-2,\dots,1$. Finally, the mean of $\hat \D_{Z,1}$ is used as an estimator for $\varphi$, and the mean of $\hat \D_{M,1}$ is taken as an estimator of $\lambda$, which are plugged into the identifying formula of Theorem \ref{theo:iden} to estimate $\theta$. (Note that the use of the $\D$-functions in this way is motivated by Lemma 1 of \cite{diaz2023efficient}.) The estimated variance of the efficient influence function (Theorem \ref{theo:eif}) is used to estimate the standard error. The procedure makes use of cross-fitting in all regression functions to ensure that the asymptotic linearity result above holds without restrictions on the complexity of the form of the machine learning estimators, thereby yielding valid statistical inference under mild conditions \citep{zheng2011cross, chernozhukov2018double}. See Section 5 of \cite{diaz2023efficient} for additional discussion of the convergence of the estimator described in that paper, which also applies to the current case.

\section{More on the structure of the observed data and intervention}\label{ssec:data}
\subsection{Measured covariates}
Baseline confounders include age, sex, race, ethnicity, body mass
index (BMI), hospital location, home oxygen status, and comorbidities
(e.g., hypertension, history of stroke, diabetes mellitus, coronary
artery disease, active cancer, cirrhosis, asthma, chronic obstructive
pulmonary disease, interstitial lung disease (ILD), HIV infection, and
immunosuppression), and are included in $L_1$.  Time-dependent
confounders $Z_t$ and $L_t$ include vital signs (e.g., highest and lowest respiratory
rate, oxygen saturation, temperature, heart rate, and blood pressure),
laboratory results (e.g., C-Reactive Protein, BUN creatinine ratio,
activated partial thromboplastin time, creatinine, lymphocytes,
neutrophils, bilirubin, platelets, D-dimer, glucose, arterial partial
pressure of oxygen, and arterial partial pressure of carbon dioxide). 

\subsection{Missingness and censoring}
In cases of
missing baseline confounders, mean substitution with an additional variable to indicate missingness is used. For missing data at
later timepoints, the last observation is carried forward. Patients
are censored at their day of hospital discharge, as AKI and vital
status were unknown after this point.  

\subsection{Time of observation}
Throughout, we have assumed that $\tau$, the length of follow-up, is equal for all observations. While this may sound like a restriction,  we do not actually assume that all units are observed for the entire length of follow-up. Censored variables are all assigned a qualitative value indicating ``censored," and variables occurring after mortality are assigned a qualitative value indicating ``mortality." For our scientific study, we have $\tau = 4$. For censored values, we adopt the procedure of including an intervention in our estimands to ``prevent censoring," which has the consequence that all outcome models at time $t$ are fitted only with units that are uncensored at time $t$.

\subsection{Censoring and survival}
Although we are mainly scientifically interested in the intervention to prevent intubation, we reiterate here that part of the intervention we consider is to ``prevent censoring"; this intervention applies to both treatment arms we are comparing (status quo and prevent intubation). What this means practically is that outcomes are estimated for censored observations based on our estimates of the nuisance functions. While this may introduce some uncertainty relative to an uncensored data set, it does not induce bias unless the censoring at random assumption is violated. While censored individuals might be thought a priori to be, for example, healthier (if they are discharged), our goal has been to control for enough baseline and time-varying covariates to make such \textit{conditional} differences minimal. Statistically, our procedure has the effect that censored observations do not provide \textit{direct} information about the effect of the intervention to prevent intubation (because they do not provide outcome data), but they do inform the estimates of nuisance functions that contribute to the estimation of the effect. \\However, we recognize that discharge could be considered a competing event and that incorporating that aspect of the problem could be a meaningful extension. Our analysis illustrates the method, but future work could extend this approach to incorporate competing risks, enhancing the interpretability of results in a substantive application. For a full-length discussion of the approach we are using for censoring, see \cite{young2020causal}.

\subsection{Survival analysis and effect parametrization}
Often, ``survival analysis" is performed by analyzing differences in hazard functions, either semi-parametrically through a proportional hazards assumption and Cox regression or through distribution-free hypothesis testing with the logrank test. Although such methods are plausible, in the context of causal inference it is more natural to define a specific parameter that quantifies the scientific effect to be estimated. In this case, the parameter is the difference in mortality rates between two different treatment regimes: one the status quo intervention, and one the intervention with prevented intubation. Thus, the form of the parameter is fundamentally no different from treatment effects in non-survival contexts except that censoring is considered. To add extra temporal dimensions to the results, one could repeat the analysis for multiple values of follow-up time  or consider alternative forms of the intervention.
\color{black}
\section{Effect modification in the AKI study}\label{ssec:modify}

Here, we aim to identify treatment effect modifiers that can be used to predict which patients are at a higher risk of a harmful effect through AKI, and which patients are likely to benefit due to all other mechanisms. 
Rigorously estimating interpretable conditional average treatment effects, or treatment effects given some covariate values, is notoriously difficult, especially when continuous covariates are considered (see for example \cite{kennedy2017non} for a review of the state-of-the-art in estimation of conditional effects). For simplicity and interpretability, we can approach this by following a strategy presented in \cite{rudolph2022ends}. Specifically, we regress the influence function (see Theorem~\ref{theo:eif}) for the interventional direct effect of intubation onto baseline covariates. (It is easily shown that the conditional (on covariates) expectation of the (uncentered) influence function for the ATE in standard settings is equal to the conditional average treatment effect; for the extension of this principle to the LMTP case see equation 6 of \cite{diaz2021nonparametricmtp}. To derive the influence function of the direct and indirect effects, note that the influence function for the difference of two parameters is equal to the difference of the influence functions of those parameters.)

Thus, we estimated univariate effect modification for each baseline variable by linear regression of the efficient influence function of the direct and indirect effect shown in equations \ref{eq:de} and \ref{eq:ie} on each variable at a time, where we use the slope of the regression as a measure of effect modification.
We show only the ten variables with the largest absolute value for univariate effect modification in Table~\ref{tab:effmod}. (If we were predominantly interested in the application of a treatment rule based on mediation effects, as described in Section \ref{sec:decisions}, then it would be more appropriate to regress the influence function onto multiple covariates simultaneously; here we focus on the simple regressions to provide a more straightforward summary.)

First, we note that this analysis is exploratory; some subsets may have relatively small sample sizes and we have not calculated standard errors. However, interstitial lung disease (ILD) is the most important variable for effect modification for both the indirect effect through AKI and the direct effect through other mechanisms. The negative effect modification parameter for ILD in Table~\ref{tab:direct} implies that preventing IMV in patients with ILD has a worse effect on survival (i.e., is more harmful/less beneficial) due to mechanisms other than AKI than in patients without ILD, whereas the positive effect modification parameter in Table~\ref{tab:indirect} means that preventing IMV in patients with ILD has a better effect on survival (i.e., is less harmful/more beneficial) due to AKI than in patients without ILD. The former effect can be explained by the fact that preventing intubation may be riskier in patients with ILD; the latter effect (smaller in size) might be due to the fact that patients with ILD are more likely to experience mortality due to lung failure rather than AKI.

Likewise, Table~\ref{tab:indirect} shows that preventing IMV in patients with cancer, chronic obstructive pulmonary disease (COPD), high lymphocyte count, costovertebral angle tenderness  (CVA), coronary artery disease (CAD), high creatinine values, and those who used home oxygen, has a better effect on survival (i.e., is less harmful/more beneficial) due to AKI than in patients without those conditions. Preventing IMV in patients who have cirrhosis or who are current smokers has a worse effect on survival (i.e., is more harmful/less beneficial) due to AKI than in patients without those conditions.

The results in Table~\ref{tab:direct} show that preventing intubation in patients who are immunosuppressed, use home oxygen, or are white, has a worse effect on survival (i.e., is more harmful/less beneficial) due to mechanisms other than AKI than in patients outside those categories. Lastly, preventing intubation in patients with high creatinine values, asthma, CVA, cirrhosis, high bilirubin, and who are former smokers, has a better effect on survival (i.e., is less harmful/more beneficial) due to mechanisms other than AKI than in patients without those conditions.

We emphasize that given the high standard errors in the effect estimates, this summary should be seen as an exploratory demonstration of our methodology; further research is needed to establish firm scientific conclusions about the effects of the timing of IMV.

\begin{table}[H]
\caption{Most important effect modifiers for the direct and indirect effect of invasive mechanical ventilation on survival through acute kidney injury, as measured by the slope of a simple linear regression of the efficient influence function. Variables with negative slopes indicate populations with lower effects (calculated as intervention minus baseline; positive values indicate that preventing intubation is relatively less harmful/more beneficial).}
\centering
\begin{minipage}{.45\textwidth}
\centering
\begin{tabular}{lr}
  \hline
Variable & Slope \\ 
  \hline
ILD & 0.34 \\ 
  Cancer & 0.32 \\ 
  COPD & 0.25 \\ 
  Lymph. count & 0.23 \\ 
  CVA & 0.16 \\ 
  CAD & 0.15 \\ 
  Current smoker & -0.13 \\ 
  Cirrhosis & -0.12 \\ 
  Creatinine & 0.10 \\ 
  Home oxygen & 0.10 \\ \hline
\end{tabular}
\subcaption{Effect modifiers for the effect of preventing IMV on survival through AKI.}\label{tab:indirect}
    \end{minipage}%
    \hspace{3mm}
    \begin{minipage}{.45\textwidth}
    \centering
\begin{tabular}{lr}
  \hline
Variable & Slope \\ 
  \hline
ILD & -1.06 \\ 
  Creatinine & 0.69 \\ 
  Asthma & 0.69 \\ 
  CVA & 0.65 \\ 
  Cirrhosis & 0.38 \\ 
  Immunosuppressed & -0.26 \\ 
  Bilirubin & 0.25 \\ 
  Former smoker & 0.22 \\ 
  Home oxygen & -0.18 \\ 
  White & -0.17 \\ 
   \hline
\end{tabular}
\subcaption{Effect modifiers for the effect of preventing IMV on survival independent of AKI.}\label{tab:direct}
\end{minipage}
\label{tab:effmod}
\end{table}

\color{black}
\section{Simulation study}

To test the accuracy of our estimator, we performed a brief Monte Carlo simulation study. For $t = 1, 2, 3$ and sample sizes $n \in \{1000, 5000\}$, we drew $J = 300$ datasets from the following data-generating mechanism, where $Y_1 = 1$.
\begin{align*}
    L_t &\sim \text{Bernoulli}(0.5) \\
    A_t &\sim \text{Bernoulli}(0.4 + 0.1\,L_t)\\
    M_t &\sim \text{Bernoulli}(0.6 - 0.1\,A_t + 0.1\,L_t)\\
    Y_{t+1} &\sim \text{Bernoulli}(Y_t\,0 + (1-Y_t)\,(0.3 + 0.05\,M_t + 0.2\,A_t - 0.001\,L_t))
\end{align*}
Note that $Y_t$ is degenerate and analogous to a survival outcome. For each dataset, we estimated $\theta(\d', \d) = \E[Y(\bar A^{\d'}, J(\bar A^\d))]$ (see main paper for an explanation of this notation) under three parameterizations of $(\d', \d)$: $\hat{\theta}(1, 1)$, $\hat{\theta}(1, 0)$, $\hat{\theta}(0, 0)$. We then estimated the randomized total, indirect, and direct effects as $\hat{\theta}(1, 1) - \hat{\theta}(0, 0)$, $\hat{\theta}(1, 1) - \hat{\theta}(1, 0)$, and $\hat{\theta}(1, 0) - \hat{\theta}(0, 0)$ respectively. The true values were estimated by drawing the counterfactuals $Y(\bar A^{\d'}, J(\bar A^\d))$ from a super-population. Estimation was performed using the \textit{lcmmtp} package in \texttt{R}; the estimator used cross-fitting ($K=5$ folds) and a generalized linear model with all possible 3-way interactions for fitting nuisance parameters. Code for the simulation is available on GitHub at \url{https://github.com/nt-williams/lcmmtp-simulation}.

\begin{table}[htbp]
\centering
\color{black}
\begin{tabular}{llcccc}
\toprule
$n$ & Effect & $\theta$ & Abs. bias & $n \cdot \text{MSE}$ & Coverage\\
\midrule
1000 & Total & -0.219 & 0.024 & 551.388 & 0.970\\
 & Indirect & 0.004 & 0.008 & 16.585 & 0.977\\
 & Direct & -0.223 & 0.032 & 503.256 & 0.967\\
\cmidrule{1-6}
5000 & Total & -0.219 & 0.002 & 4.815 & 0.937\\
 & Indirect & 0.004 & 0.000 & 0.392 & 0.923\\
 & Direct & -0.223 & 0.002 & 3.726 & 0.947\\
\bottomrule
\end{tabular}
\caption{\color{black}Summary of simulation results. The table presents the sample size ($n$), the true parameter value ($\theta$), absolute bias, the product of the sample size and mean squared error ($n \cdot \text{MSE}$), and the coverage probability for the total, indirect, and direct effect in the simulation study.}\label{tab:scenarios}
\end{table}

\section{Event counts in illustrative study}

\begin{table}[H]

\caption{\color{black}\label{tab:tab:oxygenation_levels}Distribution of Oxygenation Levels (Days 1–4)}
\centering
\begin{tabular}[t]{lrrrr}
\toprule
Day & 1 & 2 & 3 & 4\\
\midrule
No oxygen & 1267 & 470 & 188 & 102\\
Non-invasive & 1921 & 2546 & 2477 & 2184\\
Mechanical vent & 176 & 249 & 304 & 349\\
Died or discharged & 0 & 99 & 395 & 729\\
\bottomrule
\end{tabular}
\end{table}

\begin{table}[H]

\caption{\color{black}\label{tab:tab:event_counts}Event Counts by Day (AKI, Mortality, Discharged)}
\centering
\begin{tabular}[t]{lrrrr}
\toprule
Day &  1 &  2 &  3 &  4\\
\midrule
AKI & 474 & 588 & 592 & 622\\
Mortality & 24 & 72 & 59 & 72\\
Discharged & 75 & 323 & 670 & 1003\\
\bottomrule
\end{tabular}
\end{table}
\clearpage
\color{black}
\clearpage
\bibliographystyle{plainnat} \bibliography{refs}

\end{document}